\documentclass[12pt, a4paper]{article} 

\usepackage[margin=1in]{geometry} 
\usepackage{amsfonts, amscd, amssymb, mathtools, mathrsfs, dsfont, bbm, bbding} 
\usepackage[amsmath, amsthm, thmmarks]{ntheorem} 
\usepackage{graphicx, xypic, color, float} 
\usepackage{indentfirst}
\usepackage{lmodern}
\usepackage[T1]{fontenc} 
\usepackage{enumerate, listings, verbatim, paralist}
\usepackage{setspace, xspace}
\usepackage[colorlinks=true, citecolor=blue]{hyperref}
\usepackage{extarrows}
\usepackage{pdfpages}
\usepackage{ulem}

\usepackage{setspace}
\AtBeginDocument{%
\addtolength\abovedisplayskip{-0.15\baselineskip}%
\addtolength\belowdisplayskip{-0.15\baselineskip}%
\addtolength\abovedisplayshortskip{-0.15\baselineskip}%
\addtolength\belowdisplayshortskip{-0.15\baselineskip}%
}

\newtheorem{thm}{Theorem} [section]

\newtheorem{lemma}[thm]{Lemma}

\newtheorem{assmp}{Assumption}[section]
\newtheorem{remark}{Remark}[section]

\numberwithin{equation}{section} 

\renewcommand{\geq}{\geqslant}
\renewcommand{\leq}{\leqslant}

\newcommand{\citethm}[1]{Theorem \ref{#1}}

\newcommand{\citelem}[1]{Lemma \ref{#1}}

\newcommand{\citeassmp}[1]{Assumption \ref{#1}}
\newcommand{\citeremark}[1]{Remark \ref{#1}}
\newcommand{\citesec}[1]{Section \ref{#1}}

\newcommand{\opfont}{\mathbb}

\newcommand{\be}{\ensuremath{\opfont{E}}}
\newcommand{\BE}[2][]{\ensuremath{\operatorname{\opfont{E}}^{#1}\!\left[#2\right]}}

\newcommand{\R}{\ensuremath{\operatorname{\mathbb{R}}}}
\newcommand{\var}[1]{\ensuremath{\operatorname{{V}ar}\!\left(#1\right)}}

\newcommand{\dd}{\ensuremath{\operatorname{d}\! }}
\newcommand{\dt}{\ensuremath{\operatorname{d}\! t}}
\newcommand{\ds}{\ensuremath{\operatorname{d}\! s}}

\newcommand{\ddp}{\ensuremath{\operatorname{d}\! p}}

\newcommand{\setg}{\mathscr{G}}
\newcommand{\setq}{\mathscr{Q}} 
\newcommand{\setr}{\mathscr{R}}
\newcommand{\setc}{\mathscr{C}}

\newcommand{\idd}[1]{\ensuremath{\operatorname{\mathds{1}}_{#1}}}

\newcommand{\nn}{\nonumber}

\newcommand{\pa}{\pi}
\newcommand{\wa}{\beta_{\mathrm{insured}}}

\newcommand{\uinsured}{\mathcal{U}_{\mathrm{insured}}}
\newcommand{\uinsurer}{\mathcal{U}_{\mathrm{insurer}}}

\newcommand{\BV}{\ensuremath{\mathscr{E}}}
\newcommand{\einf}{\ensuremath{\mathrm{ess\:inf\:}}}
\newcommand{\esup}{\ensuremath{\mathrm{ess\:sup\:}}}

\newcommand{\barq}{\overline{Q}}
\newcommand{\barg}{\overline{G}}
\newcommand{\ep}{\varepsilon}
\newcommand{\opl}{\Phi}

\newcommand{\ms}{m_0}

\newcommand{\vare}{\theta}
\newcommand{\varc}{\sigma}

\newcommand{\ol}[1]{\mathscr{L}_{#1}}

\newcommand{\oplb}{\Psi}
\newcommand{\dsp}{\displaystyle}
\newcommand{\oplc}[1]{\mathscr{K}_{#1}}
\newcommand{\pw}{\mu}
\newcommand{\opld}{\Gamma}
\newcommand{\ople}{\Lambda}

\newcommand{\varf}{\rho}

\newcommand{\setac}{{AC}}

\begin{document}
\title{Moral-hazard-free insurance: mean-variance premium principle and rank-dependent utility theory}
\author{Zuo Quan Xu\thanks{Department of Applied Mathematics, The Hong Kong Polytechnic University, Kowloon, Hong Kong, China. Email: \url{maxu@polyu.edu.hk}. The author acknowledges financial support from The NSFC (No.11971409), The Hong Kong RGC (GRF No.15202421 and No.15202817), The PolyU-SDU Joint Research Center on Financial Mathematics, The CAS AMSS-POLYU Joint Laboratory of Applied Mathematics, and The Hong Kong Polytechnic University.}}
\date{\today}
\maketitle

\begin{abstract} 
This paper investigates a Pareto optimal insurance problem, where the insured maximizes her rank-dependent utility preference and the insurer is risk neutral and employs the mean-variance premium principle. To eliminate potential moral hazard issues, we only consider the so-called moral-hazard-free insurance contracts that obey the incentive compatibility constraint. 
The insurance problem is first formulated as a non-concave maximization problem involving Choquet expectation, then turned into a concave quantile optimization problem and finally solved by the calculus of variations method. The optimal contract is expressed by a second-order ordinary integro-differential equation with nonlocal operator. 
An effective numerical method is proposed to compute the optimal contract assuming the probability weighting function has a density. Also, we provide an example which is analytically solved. 
\end{abstract}

\section{Introduction.}

Insurance is a widely used tool that transfers a part of a risk from an innocent party (the insured) to an insurance carrier (the insurer, an insurance or reinsurance company) at the cost of paying a premium.
The well-known insurance contracts include quota share, deductible and full coverage. 
In practice, a fundamental challenge for insurers is how to design insurance contracts that achieve Pareto optimality/efficiency (PO/PE, for short) between the insurer and the insured. This paper focuses on such a risk-sharing insurance design problem. 

The diversity of insurer's and insured's premium principles and goals results in 
various PO insurance contracts.
The expected-value, standard derivation, and variance premium principles are the most widely used premium principles in the insurance study and practice. The most popular goals of the insured include maximizing the expected utility (EU) of the net wealth in the static setting, maximizing the discounted dividend payout and minimizing the probability of ruin in the dynamic setting. In this paper, we consider a static model and assume the insurer employs the mean-variance premium principle, a combination of the expected-value and variance premium principles; while the insured aims at maximizing the rank-dependent utility (RDU, developed by Quiggin \cite{Q82}) preference of her net wealth. Because the RDU preference involves a probability weighting function (also called probability distortion function), the target of the insured becomes a nonlinear Choquet expectation, making our problem being a challenging non-concave maximization problem. Besides the RDU theory, probability weighting function also plays a key role in many other behavioral theories of choice under uncertainty, such as Kahneman and Tversky's \cite{KT79, TK92} cumulative prospect theory, Yaari's \cite{Y87} dual model, Lopes' \cite{L87} SP/A model. These behavioral theories provide satisfactory answers to many paradoxes for which the EU theory fails to explain (see, e.g. Friedman and Savage \cite{FS48}, Allais \cite{A53}, Ellsberg \cite{E61}, Mehra and Prescott \cite{MP85}). In this paper, we consider general probability weighting functions which are allowed to be discontinuous. 
By contrast, many existing insurance models only consider differentiable probability weighting functions; see, Jin and Zhou \cite{JZ08}, Xu \cite{X16}, Xu, Zhou and Zhuang \cite{XZZ19}, Xu \cite{X21}, among many others. 

The present paper is a follow-up study of the author's previous work \cite{X21}. In the previous work, we assume the insurer uses the expected-value premium principle, whereas in the present model the insurer employs the more practical mean-variance premium principle. The former principle is \emph{linear} so that it is additive; by contrast, the latter is \emph{nonlinear and non-additive}, which leads to a more challenging problem. The mean-variance premium principle is a combination of the expected-value and variance premium principles. It contains each of them as a special case. Also, it takes into account the variability of the insurer's share, so it is less problematic in practice than some other premium principles such as the expected-value premium principle. See Deprez and Gerber \cite{DG85}, Kaluszka \cite{K01}, Hipp and Taksar \cite{HT10}, Yao, Yang and Wang \cite{YYW14}, Liang, Liang and Young \cite{LLY20} with mean-variance or standard deviation premium principles in either static or dynamic settings. As mentioned above, another difference between this paper and existing works is that we consider quite general probability weighting functions rather than only those absolutely continuous ones as in Xu \cite{X21}. 
The new setup not only covers more practical cases, but also brings a lot of mathematical challenges. 
Mathematically speaking, it leads to a new semi-linear second-order ordinary integro-differential equation (OIDE) with \emph{nonlocal} operator, which is not easy to solve, even numerically. If the probability weighting function has a density, we further reduce the OIDE to a well-studied ordinary differential equation (ODE) of two unknown functions with \emph{local} operator. It is an initial value ODE problem and can be numerically solved effectively. Unfortunately, we cannot find a way to do this when the probability weighting function has no density. We leave this to experts in numerical OIDEs. 

An important feature of our model (same as Xu \cite{X21}) is that we take the so-called \emph{incentive compatibility constraint} into account. In the insurance literature, many works ignore this constraint intentionally or unintentionally; see, Deprez and Gerber \cite{DG85}, Gajek and Zagrodny \cite{GZ00}, Kaluszka \cite{K01}, Liang, Liang and Young \cite{LLY20}, Guan, Xu and Zhou \cite{GXZ22}, to name a few. We believe it is mainly due to the mathematical challenges arising from the constraint to force the authors to ignore this constraint. 
In some existing models, the optimal contracts turn out to be quota share or stop loss so that the incentive compatibility constraint is automatically satisfied. But more often, in particular when probability weighting function is involved, the optimal contracts may suffer from some serious moral hazard issues such as hiding or exaggerating losses; see, e.g., Bernard et al. \cite{BHYZ15}. A large proportion of existing works do not discuss if their optimal contracts are free of moral hazard issues. 
Economically speaking, when modeling insurance problems, concerns from both the insurer and the insured should be taken into account simultaneously; so, to avoid the potential moral hazard issues, both compensation and retention functions shall be \emph{a prior} increasing\footnote{In this paper, ``increasing'' means non-decreasing and ``decreasing'' means non-increasing.} for the optimal contracts. This simple fact is called the \emph{incentive compatibility constraint} by Huberman, Mayers and Smith Jr \cite{HMS83} and Picard \cite{P00}. We call an insurance contract \emph{moral-hazard-free} if it obeys this constraint. We only focus on the moral-hazard-free contracts in this paper, so the optimal contracts automatically avoids the potential moral hazard issues.
Mathematically speaking, this constraint leads to 
the second-type quantile optimization problem defined by Xu \cite{X21}. The decision quantiles are subject to a bounded derivative constraint, an infinity-dimensional constraint. This usually leads to a double-obstacle OIDE/ODE problem. 
It is quite different from the first-type quantile optimization problem where the decision quantiles are subject to a one-side derivative constraint. If one ignores the incentive compatibility constraint when designing PO contracts, the problem reduces to a single-obstacle OIDE/ODE problem. In fact, the first-type quantile optimization problem can be solved by a simple relaxation method; please refer to Xu \cite{X16} and Hou and Xu \cite{HX16} for the latest development of the relaxation method. 

Our approach to solving the insurance design problem in this paper consists of several steps: we first transform the problem into a quantile optimization problem; and then show the latter is a concave problem; the calculus of variations method (or equivalently, the first order condition) is then applied to get an equivalent optimality condition; we further derive an equivalent condition in terms of a semi-linear second order OIDE with nonlocal operator. 
Eventually, the optimal solution to the original problem is expressed by the solution to the OIDE. In this process, we show the OIDE has an \emph{almost} classical solution. This is the best possible smoothness result, since the OIDE has no classical solution in general, which can be seen from our example in \citesec{example}. To the best of our knowledge, it is the first time that this type of OIDE with nonlocal operator appears in the insurance and financial economics literature. 

The above approach is introduced by the author in \cite{X21}. The linear expected-value premium principle is assumed for the insurer in \cite{X21}, so the quantile optimization problem is naturally a concave optimization problem. Hence, the calculus of variations method gives an equivalent optimality condition. This paper considers the nonlinear mean-variance premium principle, in order to get an equivalent optimality condition from the first order condition, we must show that the quantile optimization problem is concave. We have successfully shown this, so the approach proposed in \cite{X21} still works. 
Other methods to deal with risk-sharing problems under the incentive compatibility constraint are available in the literature. For instance, Carlier and Lachapelle \cite{CL11} use a probabilistic method to study a class of risk-sharing problems. They provide an iteration method to get the numerical solution. They show the convergence of their scheme, but do not give the speed of convergence, which seems to be a very hard problem. By contrast, our method reduces the problem to solve an OIDE/ODE problem, where the solvability is well-established in the numerical differential equation literature. The most up-to-date numerical methods to solve differential equations such as neural networks and deep learning might be applied to them as well. 

The rest of this paper is organized as follows. In \citesec{problem}, we introduce a PO insurance problem. In \citesec{quantileproblem}, the problem is turned into a quantile optimization problem via change of variables. We also show the quantile optimization problem is a concave one in this part. \citesec{solution} is devoted to solving the quantile optimization problem by the calculus of variations method. We express the optimal solution to the original problem by the solution of some OIDE/ODE. 
An example with an explicit solution is provided as well. \citesec{conclusion} concludes the paper. 
\subsection*{Notation.}
Throughout the paper, we fix an atom-less probability space. 
For any random variable $Y$ in the probability space, we denote its probability distribution function by $F_{Y}$; and define its quantile function (or the left-continuous inverse function of $F_{Y}$) by
\[F^{-1}_{Y}(p)=\inf\big\{z\in\R\;\big|\; F_{Y}(z)\geq p\big\}, \quad p\in(0, 1], \]
with the convention that $\inf\emptyset=+\infty$ and set
\[F^{-1}_{Y}(0)=\lim_{p\to 0+}F^{-1}_{Y}(p).\] By this definition, $F^{-1}_{Y}(0)=\einf Y$ and $F^{-1}_{Y}(1)=\esup Y$. So $Y$ is a bounded random variable if and only if 
\[-\infty<F^{-1}_{Y}(0)\leq F^{-1}_{Y}(1)<\infty.\] 
\par
By definition, all the quantile functions (or simply called quantiles, for short) are increasing and left-continuous.
Generally speaking, quantile functions may be discontinuous, but 
because of the presence of the incentive compatibility constraint in our model, we will face absolutely continuous quantiles only in this paper. This will simplify our arguments in many circumstances. 
\par
For $p\geq 1$, let $L^p([0, 1])$ be the set of measurable functions $f: [0, 1]\to \R$ such that \[\int_0^1 |f(t)|^p\dt<\infty.\] 
Let $\setac([0, 1])$ be the set of absolutely continuous functions $f: [0, 1]\to \R$.
Let $C^{2-}([0, 1])$ be the set of functions $f: [0, 1]\to \R$ that are differentiable with derivatives $f'\in \setac([0, 1])$. Clearly $C^{2}([0, 1])\subseteq C^{2-}([0, 1])\subseteq C^{1}([0, 1])$. 
\par
In what follows, ``almost everywhere'' (a.e.) and ``almost surely'' (a.s.) may be suppressed for notation simplicity in most circumstances if no confusion would occur. 

\section{Problem formulation.}\label{problem}

In the Pareto optimal (also called Pareto efficient) insurance problem, one seeks the best way to share a potential loss between an insurer (``He'', an insurance or reinsurance company) and an insured (``She'') to achieve Pareto optimality/efficiency.
\par
We use the same notation as in Xu \cite{X21}. We use a random variable $X\geq 0$ to represent the potential loss covered by the insurance contract.
Let $I(x)$ and $R(x)$ be the loss borne by the insurer and the insured when a real loss $X=x$ occurs. They are respectively called the \emph{compensation} (also called \emph{indemnity}) and \emph{retention} functions in the insurance literature. 
A contract is called full coverage if $I(x)\equiv x$; called deductible (with deductible $d$) if $I(x)\equiv \max\{x-d, 0\}$. Economically speaking, both the insurer and the insured should bear a greater financial responsibility when a larger loss would occur, otherwise moral hazard issue may arise (see more discussions in Bernard et al. \cite{BHYZ15} and Xu, Zhou and Zhuang \cite{XZZ19}). This is called the \emph{incentive compatibility constraint} by Huberman, Mayers and Smith Jr \cite{HMS83} and Picard \cite{P00}. 
Mathematically speaking, because $I(x)+R(x)\equiv x$ and both $I$ and $R$ are increasing, 
the set of acceptable compensations is\footnote{We refer to Xu \cite{X21} for a detailed discussion. } 
\begin{align*} 
\setc&=\Big\{I:[0, \infty)\to [0, \infty)\; \big|\; \mbox{$I$ is absolutely }\nn\\
&\qquad\quad~\mbox{ continuous with $I(0)=0$ and $0\leq I'\leq 1$ a.e.}\Big\}, 
\end{align*} 
and the set of acceptable retentions is 
\begin{align*} 
\setr 
&=\Big\{R:[0, \infty)\to [0, \infty)\; \big|\; \mbox{$R$ is absolutely }\nn\\
&\qquad\quad~\mbox{ continuous with $R(0)=0$ and $0\leq R'\leq 1$ a.e.}\Big\}.
\end{align*} 
Clearly $\setr=\setc$. We call the compensations in $\setc$ and the retentions in $\setr$ \emph{moral-hazard-free} and will only consider them in the sequel.
\par
An insurance contract is a pair $(\pa, I)$, where $\pa\in\R$ is a premium that the insured pays to the insurer at initial time and $I$ is a moral-hazard-free compensation in $\setc$.
We assume the insurer uses the mean-variance premium principle
\begin{align} 
\uinsurer(\pa, I)=\pa-\vare\BE{I(X)}-\varc\var{I(X)}, 
\end{align} 
where $\vare$ and $\varc$ are nonnegative constants. In practice, there is usually a safety loading $\vare$ for the insurer (see, e.g., Daykin et al., \cite{DPP94}), and $\varc$ is used to control the volatility of the insurer's share. When $\varc=0$, the mean-variance premium principle reduces to the classical expected-value premium principle that is considered by Xu \cite{X21}, we will not investigate this case again, hence assume $\varc>0$ from now on. 
Meanwhile, we assume the insured evaluates contracts by the rank-dependent utility preference 
\begin{align} \label{uinsured}
\uinsured(\pa, I)=\BV\Big(u\big(\wa-\pa-X+I(X)\big)\Big).
\end{align} 
Here the constant $\wa$ stands for the final wealth of the insured, so $\wa-\pa-X+I(X)$ is the insured's net wealth after deducting claims. The insured's utility function $u$ is concave, strictly increasing and differentiable on $\R$, which implies $u'$ is a continuous positive function. 
The expectation $\BV$ for a random variable $Y$ is defined as 
\begin{align}\label{vdef}
\BV(Y)=\int_{[0, 1]} F_{Y}^{-1}(p) \pw(\ddp), 
\end{align} 
where $\pw$ is a (probability) measure on $[0, 1]$ with $\pw(\{0\})=0$. In our presentation below, $Y$ will represent bounded random variables, so $\BV(Y)$ are always well-defined. 
The expectation $\BV$ is nonlinear (indeed it is a Choquet expectation) except for the trivial case $ \pw(\ddp)=\ddp$ where $\BV$ reduces to the classical linear mathematical expectation $\be$. 
\begin{remark}
In He et al. \cite{HJZ15}, $-\BV(Y)$ is called 
the weighted VaR (WVaR) risk measure for $Y$, which is a generalization of Value-at-Risk (VaR) and Expected Shortfall (ES), and encompasses many well-known risk measures that are widely used in finance and actuarial sciences, such as spectral risk measures and distortion risk measures; see Wei \cite{W18} for a review.
\end{remark}

\begin{remark}\label{remark:bhm}
In many existing works such as \cite{X16}, \cite{XZZ19}, \cite{X21}, the expectation $\BV$ is defined as 
\begin{align*} 
\BV(Y)=\int_{[0,\infty)}z\dd \big(1-w(1-F_{Y}(z)) \big),
\end{align*} 
where $w$ is a probability weighting function that is absolutely continuous, increasing and one-to-one maps $[0,1]$ to itself. By change of variable, we get
\begin{align*} 
\BV(Y)=\int_{[0, 1]} F^{-1}_{Y}(z)w'(1-p)\ddp. 
\end{align*} 
Hence, the probability measure $\pw$ in \eqref{vdef} in this case is given by \[\dd\pw=w'(1-p)\ddp.\] When $w$ is not absolutely continuous, we can take 
\[\mu([0,p])=1-w(1-p).\]
\end{remark}

\par
A contract $({\pa}^*, {I}^*)$ is called Pareto optimal/efficient if there is no other one feasible contract $(\pa, I)$ such that 
\[\uinsured(\pa, I)\geq \uinsured({\pa}^*, {I}^*), \quad \uinsurer(\pa, I)\geq \uinsurer({\pa}^*, {I}^*)\] 
and 
\[\uinsured(\pa, I)+\uinsurer(\pa, I)> \uinsured({\pa}^*, {I}^*)+\uinsurer({\pa}^*, {I}^*).\] 
In other words, it is impossible to increase one of the insurer's and the insured's valuations for a PO contract without reducing the other one. All the PO contracts form a set, called the Pareto frontier. 

A contract $({\pa}^*, {I}^*)$ is PO if and only if there exists a $\gamma\in\R$ such that $({\pa}^*, {I}^*)$ is an optimal solution to the problem 
\begin{align*} 
\sup_{\pa\in\R, \; I\in\setc}&\quad \uinsured(\pa, I) \\
\mathrm{s.t.\ \ }&\quad \uinsurer(\pa, I)\geq \gamma. 
\end{align*} 
Under our specific setting, the above becomes 
\begin{align*} 
\sup_{\pa\in\R, \; I\in\setc}&\quad \BV\Big(u\big(\wa-\pa-X+I(X)\big)\Big) \\
\mathrm{s.t.\ \ }&\quad \pa-\vare\BE{I(X)}-\varc\var{I(X)}\geq \gamma. 
\end{align*} 
Notice the objective $\BV\Big(u\big(\wa-\pa-X+I(X)\big)\Big)$ is decreasing in $\pa$, so any PO contract $(\pa^*, I^*)$ shall make the constraint tight, namely
\[\pa^*-\vare\BE{I^*(X)}-\varc\var{I^*(X)}=\gamma.\]
Therefore, by removing $\pa^*$ from the above problem, it suffices to study the problem 
\begin{align} \label{opi}
\sup_{I\in\setc}&\quad \BV\Big(u\big(\wa-\gamma-\vare\BE{I(X)}-\varc\var{I(X)}-X+I(X)\big)\Big). 
\end{align} 
If $I^*_{\gamma}$ is an optimal solution to the above problem \eqref{opi}, then 
\[\big(\gamma+\vare\BE{I^*_{\gamma}(X)}+\varc\var{I^*_{\gamma}(X)}, \; I^*_{\gamma}\big)\]
is a PO contract. Conversely, every PO contract is of the above form with certain $\gamma$. Hence, it suffices to solve the problem \eqref{opi}. Without confusion, we also call its solution (which is indeed an optimal compensation) a PO moral-hazard-free contract. In the same spirit, a PO contract is called deductible if the compensation in the contract is a deductible one. 

\par
Following Xu \cite{X21}, we put the following technical assumptions on $X$ throughout the paper. 
\begin{assmp}\label{ass1}
The quantile function $F_{X}^{-1}$ of the potential loss $X$ satisfies $F_{X}^{-1}(0)=\einf X=0$ and $F_{X}^{-1}(1)=\esup X<\infty$. Furthermore, $F_{X}^{-1}\in\setac([0, 1])$ and $\big(F_{X}^{-1}\big)'(p)> 0$ for a.e. $p\in(\ms, 1)$, where $\ms=F_X(0)<1$. 
\end{assmp} 
If $\ms>0$, then $X$ has a positive mass at 0, so \citeassmp{ass1} covers the most common and important case with loss having a positive mass at 0 in insurance practice. Under \citeassmp{ass1}, we only need to deal with bounded random variables throughout this paper, which will simplify our subsequent arguments. Remark that our method can be also applied to the case of unbounded potential loss $X$, but it requires more careful mathematical derivations such as on integrations. This is out of the main goal of this paper, so we leave it to the interested readers. 

Under \citeassmp{ass1}, the probability distribution function $F_X$ is continuous on $[0, 1]$ and strictly increasing on $[\ms, 1]$. Moreover, we have $F_{X}^{-1}(F_X(x))=x$ for all $\einf X\leq x\leq \esup X$, $F_{X}^{-1}(p)=0$ for $p\leq \ms$ and $F_{X}^{-1}(p)>0$ for $p>\ms$. These facts will be used frequently in the subsequent analysis without claim. 
\par

\section{Quantile optimization problem.}\label{quantileproblem}
Generally speaking, the probability measure $\pw$ in \eqref{vdef} makes the preference $\BV$ a \emph{nonlinear} expectation (in fact, it is a Choquet expectation), so the problem \eqref{opi} is a challenging non-concave optimization problem. To tackle the problem \eqref{opi}, we use the so-called \emph{quantile optimization method}; see \cite{CDT00,DS07,CD08, JZ08, HZ11,XZ13, BHYZ15, X16, HX16,XZ16, W18, XZZ19, X21, MX21}
for the recent development of this method. 
\par
Our approach was introduced by Xu \cite{X16}. 
Following Xu \cite{X16}, we first make change of variables to find an equivalent quantile optimization problem to the problem \eqref{opi}, then use the quantile optimization 
techniques to study it, and finally recover the optimal solution to the problem \eqref{opi}. 
\par 
Because our probability space is atom-less, there exists a random variable $U$, which is uniformly distributed on $(0, 1)$, such that $X=F_{X}^{-1}(U)$ almost surely (see, e.g. Xu \cite{X14}). 
For $R\in\setr$, let 
\begin{align} \label{GR}
G(p)=R\big(F_{X}^{-1}(p)\big), \quad p\in[0, 1].
\end{align}
Then $G$ is an increasing function and satisfies 
\begin{align} \label{G(u)}
G(U)=R\big(F_{X}^{-1}(U)\big)=R(X)=X-I(X).
\end{align} 
Furthermore, using the fact $F_{X}^{-1}(F_X(x))=x$, 
\begin{align} \label{RG}
R(x)=R\big(F_{X}^{-1}(F_X(x))\big)=G(F_X(x)), \quad x\in
\big[\einf X, \; \esup X\big]. 
\end{align}
Writing 
\begin{align*} 
Y&=\wa-\gamma-\vare\BE{I(X)}-\varc\var{I(X)}-X+I(X)\\
&=\wa-\gamma-\vare\BE{I(X)}-\varc\var{I(X)}-G(U), 
\end{align*} 
it is not hard to verify that the quantile function of $Y$ is given by 
\[F_{Y}^{-1}(p)=\wa-\gamma-\vare\BE{I(X)}-\varc\var{I(X)}-G(1-p), \quad \mbox{a.e.}\; p\in[0, 1].\]
Inserting it into \eqref{uinsured}, we get 
\begin{align*} 
&\quad\;\:\BV\Big(u\big(\wa-\gamma-\vare\BE{I(X)}-\varc\var{I(X)}-X+I(X)\big)\Big)\\
&=\int_{[0, 1]}u\big(\wa-\gamma-\vare\BE{I(X)}-\varc\var{I(X)}-G(1-p)\big) \pw(\ddp).
\end{align*} 
Thanks to \eqref{G(u)}, we have 
\[\BE{I(X)}=\BE{X-G(U)}=\BE{X}-\int_{0}^{1}G(t)\dt\]
and 
\begin{align*} 
\var{I(X)}
&=\var{X-G(U)}\\
&=\var{G(U)}-2\BE{XG(U)}+2\BE{X}\BE{G(U)}+\var{X}\\
&=\BE{G(U)^2}-\big(\BE{G(U)}\big)^2-2\BE{F_{X}^{-1}(U)G(U)}\\
&\qquad+2\BE{X}\BE{G(U)}+\var{X}\\
&=\int_{0}^{1}G(t)^2\dt-\Big(\int_{0}^{1}G(t)\dt\Big)^2-2\int_{0}^{1} F_{X}^{-1}(t)G(t)\dt\\
&\qquad+2\BE{X} \int_{0}^{1}G(t)\dt+\var{X}, 
\end{align*} 
so 
\begin{multline*} 
\qquad\BV\Big(u\big(\wa-\gamma-\vare\BE{I(X)}-\varc\var{I(X)}-X+I(X)\big)\Big) \\
=\int_{[0, 1]}u\big(\ol{G}-G(1-p)\big) \pw(\ddp).\qquad
\end{multline*}
where the operator $\ol{}: L^2([0, 1])\to \R$ is defined as 
\begin{align*} 
\ol{f(\cdot)}&=\varc\Big(\int_{0}^{1}f(t)\dt\Big)^2-\varc\int_{0}^{1}f(t)^2\dt+2\varc\int_{0}^{1}F_{X}^{-1}(t)f(t)\dt\\
&\qquad\;+(\vare-2\varc\BE{X})\int_{0}^{1}f(t)\dt+\wa-\gamma-\vare\BE{X}-\varc\var{X}.
\end{align*} 
For any constant $c$, we have
\begin{align}\label{olproperty}
\ol{f(\cdot)+c}&=\ol{f(\cdot)}+c\vare. 
\end{align} 
\par
We now rewrite the compatibility constraint on $R\in\setr$ in terms of the new decision variable $G$. It is not hard to show that $R\in\setr$ if and only if $G\in\setg$,\footnote{For more details we refer to Xu, Zhou and Zhuang \cite{XZZ19}.} where 
\begin{align*}
\setg&=\Big\{G:[0, 1]\to [0, \infty)\;\big|\; \mbox{$G$ is absolutely} \nn\\
&\qquad\quad~\mbox{ continuous with $G(0)=0$ and $0\leq G'\leq h$ a.e.}\Big\}, \label{setg}
\end{align*} 
and 
\begin{align} \label{def:h}
h(p)=\left(F_{X}^{-1}\right)'(p)\geq 0, \quad \mbox{a.e.}\;p\in[0, 1].
\end{align} 
Thanks to \citeassmp{ass1}, 
\begin{align} \label{h1}
\int_0^1 h(t)\dt=F_{X}^{-1}(1)=\esup X<\infty.
\end{align} 
As a consequence, we have 
\begin{align} \label{boundedsetg}
0\leq G\leq \esup X,\quad G\in\setg. 
\end{align} 

The preceding change of variables reduces the optimization problem \eqref{opi} under compatibility constraint to the following second-type quantile optimization problem 
\begin{align} \label{opi2}
\dsp\sup_{G\in\setg}&\; \int_{[0, 1]}u\big(\ol{G}-G(1-p)\big) \pw(\ddp).
\end{align}
At first sight, because of the nonlinear term $ \varc\big(\int_{0}^{1}G(t)\dt\big)^2-\varc\int_{0}^{1}G(t)^2\dt$ in $\ol{G}$, it seems that the problem \eqref{opi2} is not a concave optimization problem. But it is indeed a concave optimization problem, which will be shown by the following lemma. 

\begin{lemma}\label{concaveol}
Suppose $0<\ep<1$ and $f_1$, $f_2\in L^2([0, 1])$. Then 
\[\ep\ol{f_1(\cdot)}+(1-\ep)\ol{f_2(\cdot)}\leq \ol{\ep f_1(\cdot)+(1-\ep)f_2(\cdot)}.\]
Moreover, the identity holds if and only if $f_1-f_2$ is a constant function in $L^2([0, 1])$. 
\end{lemma}
\begin{proof} 
Clearly, 
\begin{align*}
&\quad\; \ol{\ep f_1(\cdot)+(1-\ep)f_2(\cdot)}- \ep\ol{f_1(\cdot)}-(1-\ep)\ol{f_2(\cdot)}\\
&=\varc\bigg[\Big(\int_{0}^{1}\ep f_1(t)+(1-\ep)f_2(t)\dt\Big)^2-\int_{0}^{1}\Big(\ep f_1(t)+(1-\ep)f_2(t)\Big)^2\dt\bigg]\\
&\qquad\quad-\varc\ep\bigg[\Big(\int_{0}^{1} f_1(t)\dt\Big)^2-\int_{0}^{1} f_1(t)^2\dt\bigg]-\varc(1-\ep)\bigg[\Big(\int_{0}^{1}f_2(t)\dt\Big)^2-\int_{0}^{1}f_2(t)^2\dt\bigg]\\
&=\varc\ep(1-\ep)\bigg[\int_{0}^{1}f_1(t)^2\dt-\Big(\int_{0}^{1}f_1(t)\dt\Big)^2+\int_{0}^{1}f_2(t)^2\dt-\Big(\int_{0}^{1}f_2(t)\dt\Big)^2\\
&\qquad\qquad\qquad+2\int_{0}^{1}f_1(t)\dt \int_{0}^{1}f_2(t)\dt-2\int_{0}^{1}f_1(t)f_2(t)\dt\bigg]\\
&=\varc\ep(1-\ep)\bigg[\int_{0}^{1}\big(f_1(t)-f_2(t)\big)^2\dt-\Big(\int_{0}^{1}\big(f_1(t)-f_2(t)\big)\dt\Big)^2\bigg]\\
&=\varc\ep(1-\ep)\bigg[\int_{0}^{1}\bigg(f_1(t)-f_2(t)-\int_{0}^{1}\big(f_1(s)-f_2(s)\big)\ds\bigg)^2\dt\bigg]\\
&\geq 0.
\end{align*}
The above inequality becomes an equation if and only if $$f_1(t)-f_2(t)=\int_{0}^{1}\big(f_1(s)-f_2(s)\big)\ds, \quad \mbox{ for a.e. $t\in[0, 1]$, }$$ which is equivalent to saying that $f_1-f_2$ is a constant function in $L^2([0, 1])$. This completes the proof. 
\end{proof}

Thanks to the above lemma and the concavity and monotonicity of $u$, we see that 
\begin{align}\label{convexity}
&\quad\; \ep\int_{[0, 1]}u\big(\ol{G_1(\cdot)}-G_1(1-p)\big) \pw(\ddp)+(1-\ep)\int_{[0, 1]}u\big(\ol{G_2(\cdot)}-G_2(1-p)\big) \pw(\ddp)\nn\\
&\leq \int_{[0, 1]}u\big(\ep\ol{G_1(\cdot)}-\ep G_1(1-p)+(1-\ep)\ol{G_2(\cdot)}-(1-\ep) G_2(1-p)\big) \pw(\ddp)\nn\\
&\leq \int_{[0, 1]}u\big(\ol{\ep G_1(\cdot)+(1-\ep)G_2(\cdot)}-(\ep G_1(1-p)+(1-\ep) G_2(1-p))\big) \pw(\ddp).
\end{align}
This shows that \eqref{opi2} is a concave optimization problem, which is expected to be easier to study than the non-concave optimization problem \eqref{opi}. 

Since \eqref{opi2} is a concave optimization problem, the calculus of variations method (or equivalently, the first-order condition) will provide not only a necessary and but also a sufficient optimality condition. Without concavity, we may not be able to deduce an equivalent optimality condition by the first order condition. 
\par
The next result is about the existence and uniqueness of the solution to the problem \eqref{opi2}. 
\begin{lemma}\label{lemexist1}
The problem \eqref{opi2} admits a unique optimal solution. 
\end{lemma}

\begin{proof}
Thanks to \eqref{boundedsetg}, all the admissible solutions to the problem \eqref{opi2} are uniformly bounded, hence the optimal value is finite. 
Therefore, there exists a sequence of admissible solutions $\{G_{n}\}_n$ to the problem \eqref{opi2} such that 
\begin{align*} 
\int_{[0, 1]}u\big(\ol{G_{n}}-G_{n}(1-p)\big) \pw(\ddp)
>\sup_{G\in\setg}\; \int_{[0, 1]}u\big(\ol{G}-G(1-p)\big) \pw(\ddp)-\frac{1}{n}.
\end{align*}
For any $n$ and $0\leq p_{1}<p_{2}\leq 1$, we have 
\[|G_{n}(p_{2})-G_{n}(p_{1})|=\int_{p_{1}}^{p_{2}}G_{n}'(t)\dt\leq \int_{p_{1}}^{p_{2}}h(t)\dt, \]
so, by virtue of \eqref{h1}, the sequence $\{G_{n}\}_n$ is uniformly equicontinuous. By the Arzel\`a-Ascoli theorem, we conclude $\{G_{n}\}_n$ has a subsequence (still denoted by $\{G_{n}\}_n$) that converges uniformly to some $\barg$. It is easy to verify that $\barg\in\setg$. By the dominated convergence theorem, 
\begin{align*} 
\int_{[0, 1]}u\big(\ol{\barg}-\barg(1-p)\big) \pw(\ddp)
&=\lim_{n\to\infty}\int_{[0, 1]}u\big(\ol{G_{n}}-G_{n}(1-p)\big)\pw(\ddp)\\
&\geq\dsp\sup_{G\in\setg}\; \int_{[0, 1]}u\big(\ol{G}-G(1-p)\big) \pw(\ddp).
\end{align*}
This shows that $\barg$ is an optimal solution to \eqref{opi2}. 

To prove the uniqueness, we suppose, on the contrary, that the problem \eqref{opi2} has two different optimal solutions $\barg_1$ and $\barg_2$. 
Notice $\barg_1(0)-\barg_2(0)=0$ and $\barg_1-\barg_2$ is continuous but not identical to zero, so $\barg_1-\barg_2$ is not a constant function in $L^2([0, 1])$. Hence, by \citelem{concaveol}, 
\[\ep\ol{\barg_1}+(1-\ep)\ol{\barg_2}< \ol{\ep \barg_1+(1-\ep)\barg_2}, \] 
for any $0<\ep<1$. 
Because $u$ is strictly increasing, the last inequality in \eqref{convexity} is strict, giving 
\begin{multline*} 
\ep\int_{[0, 1]}u\big(\ol{\barg_1}-\barg_1(p)\big) \pw(\ddp)+(1-\ep)\int_{[0, 1]}u\big(\ol{\barg_2}-\barg_2(p)\big)\pw(\ddp)\nn\\
<\int_{[0, 1]}u\big(\ol{\ep \barg_1+(1-\ep)\barg_2}-(\ep \barg_1(1-p)+(1-\ep) \barg_2(1-p))\big) \pw(\ddp). 
\end{multline*}
This clearly contradicts the optimality of $\barg_1$ and $\barg_2$ and confirms the uniqueness. 
\end{proof}

\par

\section{Optimal solution.}\label{solution}
As the problem \eqref{opi2} is a concave optimization problem, we can apply the calculus of variations method to solve it. This method leads to the following result, which completely characterizes the unique optimal solution to the problem \eqref{opi2}. 
\begin{lemma}[Optimality condition I]\label{op1}
Suppose $\barg\in\setq$. Then $\barg$ is the optimal solution to the problem \eqref{opi2} if and only if it satisfies 
\begin{align} \label{optimalcondition}
&\int_{[0, 1]}\Big[ \int_{0}^{1}\Big( 2\varc \int_{0}^{1}\barg(t)\dt -2\varc \barg(t)+2\varc F_{X}^{-1}(t)+\vare-2\varc\BE{X} \Big)(G(t)-\barg(t))\dt \nn\\
&\quad\qquad\quad- (G(1-p)-\barg(1-p))
\Big]u'\big(\ol{\barg(\cdot)}-\barg(1-p)\big) \pw(\ddp)\leq 0 \quad \mbox{for any $G\in\setq$.}\quad
\end{align} 
\end{lemma}

\begin{proof}
Suppose $\barg$ is the optimal solution to the problem \eqref{opi2}. For any $G\in\setq$, $\ep\in(0, 1)$, define 
\[G_{\ep}(p)=\barg(p)+\ep (G(p)-\barg(p)), \quad p\in[0, 1].\]
Then $G_{\ep}\in\setg$. 
Because $\barg$ is the optimal solution to the problem \eqref{opi2}, applying Fatou's lemma, we get 
\begin{align*} 
0 &\geq \liminf_{\ep\to 0+}\frac{1}{\ep}\bigg[\int_{[0, 1]}u\big(\ol{G_{\ep}(\cdot)}-G_{\ep}(1-p)\big)-u\big(\ol{\barg(\cdot)}-\barg(1-p)\big) \pw(\ddp)\bigg]\\
&\geq \int_{[0, 1]} \liminf_{\ep\to 0+}\frac{1}{\ep}\bigg[u\big(\ol{G_{\ep}(\cdot)}-G_{\ep}(1-p)\big)-u\big(\ol{\barg(\cdot)}-\barg(1-p)\big)\bigg] \pw(\ddp) \\
&=\int_{[0, 1]} \Big(u\big(\ol{G_{\ep}(\cdot)}-G_{\ep}(1-p)\big)\Big)'\Big|_{\ep=0} \pw(\ddp).
\end{align*} 
A simple calculation shows that the last integrand is equal to the integrand of \eqref{optimalcondition}. So we proved \eqref{optimalcondition}. 
\par
On the other hand side, suppose $\barg\in\setq$ and it satisfies \eqref{optimalcondition} but is not optimal to the problem \eqref{opi2}. Then there exist a function $G_{1}\in\setq$ and a constant $c>0$ such that 
\[\int_{[0, 1]}u\big(\ol{G_{1}(\cdot)}-G_{1}(1-p)\big) \pw(\ddp)>\int_{[0, 1]}u\big(\ol{\barg(\cdot)}-\barg(1-p)\big) \pw(\ddp)+c.\]
For $\ep\in(0, 1)$, let 
\[G_{\ep}(p)=\ep G_1(p)+(1-\ep)\barg(p), \quad p\in[0, 1].\]
By virtue of \eqref{convexity}, we have 
\begin{align*} 
&\quad\;\int_{[0, 1]}u\big(\ol{G_{\ep}(\cdot)}-G_{\ep}(1-p)\big) \pw(\ddp) \\
&\geq \ep\int_{[0, 1]}u\big(\ol{G_{1}(\cdot)}-G_{1}(1-p)\big) \pw(\ddp)+(1-\ep)\int_{[0, 1]}u\big(\ol{\barg(\cdot)}-\barg(1-p)\big) \pw(\ddp)\\
&\geq \int_{[0, 1]}u\big(\ol{\barg(\cdot)}-\barg(1-p)\big) \pw(\ddp)+c\ep, 
\end{align*} 
so
\begin{align*} 
\liminf_{\ep\to 0+}\frac{1}{\ep}\Big[\int_{[0, 1]}u\big(\ol{G_{\ep}(\cdot)}-G_{\ep}(1-p)\big)-u\big(\ol{\barg(\cdot)}-\barg(1-p)\big) \pw(\ddp)\Big]\geq c>0.
\end{align*} 
But the dominated convergence theorem and \eqref{optimalcondition} lead to 
\begin{align*} 
&\quad\;\liminf_{\ep\to 0+}\frac{1}{\ep}\Big[\int_{[0, 1]}u\big(\ol{G_{\ep}(\cdot)}-G_{\ep}(1-p)\big)-u\big(\ol{\barg(\cdot)}-\barg(1-p)\big) \pw(\ddp)\Big]\\
&=\int_{0}^{1} \liminf_{\ep\to 0+}\frac{1}{\ep}\Big[u\big(\ol{G_{\ep}(\cdot)}-G_{\ep}(1-p)\big)-u\big(\ol{\barg(\cdot)}-\barg(1-p)\big)\Big] \pw(\ddp) \\
&=\int_{[0, 1]} \Big(u\big(\ol{G_{\ep}(\cdot)}-G_{\ep}(1-p)\big)\Big)'\Big|_{\ep=0} \pw(\ddp)\\
&\leq 0, 
\end{align*}
contradicting the above inequality. This completes the proof. 
\end{proof}
\par
By this result, we see that solving the problem \eqref{opi2} reduces to the problem of finding a $\barq\in\setq$ that satisfies the condition \eqref{optimalcondition}. But one cannot find such a $\barq$ easily from \eqref{optimalcondition}, because it requires to compare $\barq$ with all the quantiles in $\setq$. Intuitively speaking, this does not reduce the difficulty of solving the problem \eqref{opi2}. 
\par

Our next goal is to find an equivalent condition to \eqref{optimalcondition} that can be easily verified. To this end, we define a function 
\begin{align*}
\opl(p)=\frac{\dsp\int_{(1-p, 1]}u'\big(\ol{\barg(\cdot)}-\barg(1-t)\big) \pw(\dt)}{\dsp\int_{[0, 1]}u'\big(\ol{\barg(\cdot)}-\barg(1-t)\big) \pw(\dt)}, \quad p\in[0, 1].
\end{align*} 
Thanks to $\pw(\{0\})=0$, one can see $\opl$ is a probability distribution function. 
Also, we define 
\begin{align*} 
\oplb(p)&=\Big(2\varc \int_{0}^{1}\barg(t) \dt+\vare-2\varc\BE{X}\Big)(p-1)\\
&\qquad+2\varc\int_{p}^{1} \big(\barg(t)-F_{X}^{-1}(t)\big)\dt+1, \quad p\in[0, 1].
\end{align*} 
Then it is easy to verify $\oplb(0)=1-\vare$ and $\oplb(1)=1$. 

In terms of these new notations, the inequality in \eqref{optimalcondition} can be written as 
\begin{align*} 
\int_{[0, 1]}\Big[\int_{0}^{1}\oplb'(t) (G(t)-\barg(t)) \dt-(G(1-p)-\barg(1-p))\Big]\dd \big(1-\opl(1-p)\big)\leq 0, 
\end{align*} 
that is, 
\begin{align*} 
\int_{0}^{1}\oplb'(t) (G(t)-\barg(t)) \dt-\int_{[0, 1]}\big(G(1-p)-\barg(1-p)\big)\dd \big(1-\opl(1-p)\big)\leq 0.
\end{align*} 
Applying integration by parts to the second integral, the above becomes 
\begin{align*} 
\int_{0}^{1}\oplb'(t)\big(G(t)-\barg(t)\big) \dt
+\int_{0}^{1}(\opl(p)-1) \big(G'(p)-\barg'(p)\big)\ddp\leq 0, 
\end{align*} 
thanks to $\opl(1)=1$, $G(0)=0$ and $\barg(0)=0$. 
By virtue of $G(0)=0$, $\barg(0)=0$ and $\oplb(1)=1$, and applying integration by parts to the first integral in above, it becomes
\begin{align*} 
\int_{0}^{1}(1-\oplb(t))\big(G'(t)-\barg'(t)\big)\dt
+\int_{0}^{1} (\opl(p)-1) \big(G'(p)-\barg'(p)\big)\ddp\leq 0, 
\end{align*} 
or 
\begin{align} \label{optimalcondition1-1}
\int_{0}^{1} \big(\opl(p)-\oplb(p)\big)\big(G'(p)-\barg'(p)\big)\ddp\leq 0. 
\end{align} 
Because $G', \;\barg'\in[0, h(p)]$, we conclude that $\barg$ satisfies 
\begin{align} \label{optimalcondition2}
\begin{cases}
\barg(0)=0;\\
\barg'(p)=h(p), &\quad\text{if } \oplb(p)-\opl(p)<0;\\
\barg'(p)\in[0, h(p)], &\quad\text{if } \oplb(p)-\opl(p)=0;\\
\barg'(p)=0, &\quad\text{if } \oplb(p)-\opl(p)>0, 
\end{cases}\; \mbox{ for a.e. $p\in[0, 1]$.}
\end{align} 
The preceding arguments are reversible, so \eqref{optimalcondition2} is equivalent to \eqref{optimalcondition}. 
The key point is that the condition \eqref{optimalcondition2} is much easier to verify than \eqref{optimalcondition} since it only depends on $\barg$ itself. 
\par
Although \eqref{optimalcondition2} is easier to verify, it is still uneasy to find or compute $\barg$ from it. We now express the condition \eqref{optimalcondition2} through an ordinary integro-differential equation by virtue of the following technical lemma. 
This OIDE can be further reduced to an ordinary differential equation later in some special cases. 
\begin{lemma}[Lemma 4.2, \cite{X21}]\label{threecases}
Suppose $a, b, c, d$ are real quantities with $b\leq c$. Then 
\[\min\{\max\{a-c, \; d\}, \;a-b\}=0\] if and only if 
\begin{align*}
\begin{cases}
a=c, &\text{if } d<0;\\
a\in[b, c], &\text{if } d=0;\\
a=b, &\text{if } d>0.
\end{cases}
\end{align*}
\end{lemma}

\begin{lemma}[Optimality condition II] \label{op2}
Suppose $\barg: [0, 1]\to\R$ is an absolutely continuous function. Then $\barg$ is the optimal solution to the problem \eqref{opi2} if and only if it satisfies $\barg(0)=0$ and the following OIDE
\begin{align} \label{vi001}
\min\Big\{\max\big\{\barg'(p)-h(p), \; \oplb(p)-\opl(p)\big\}, \;\barg'(p)\Big\}=0, \quad\mbox{a.e.}\;p\in[0, 1], 
\end{align} 
where
\begin{align} \label{opldef}
\opl(p)=\frac{\dsp\int_{(1-p, 1]}u'\big(\ol{\barg(\cdot)}-\barg(1-t)\big) \pw(\dt)}{\dsp\int_{[0, 1]}u'\big(\ol{\barg(\cdot)}-\barg(1-t)\big) \pw(\dt)}, 
\end{align}
and
\begin{align} \label{oplbdef}
\oplb(p)&=\Big(2\varc \int_{0}^{1}\barg(t) \dt+\vare-2\varc\BE{X}\Big)(p-1)+2\varc\int_{p}^{1} \big(\barg(t)-F_{X}^{-1}(t)\big)\dt+1.
\end{align} 
\end{lemma}

\begin{proof}
This is an immediate consequence of the optimality condition \eqref{optimalcondition2} and \citelem{threecases}. 
\end{proof}

\par
There are three unknown functions $\barg$, $\opl$ and $\oplb$ in \eqref{vi001}, so it is not easy to solve. 
We want to further simplify \eqref{vi001}. 
To this end, define an operator 
\begin{align}
\dsp\oplc{f(\cdot)}(p)=\frac{\int_{(1-p, 1]}u'\big(\ol{F_{X}^{-1}(\cdot)+\frac{1}{2\varc}( f(0)- f(\cdot))}
-F_{X}^{-1}(1-t)-\tfrac{1}{2\varc}( f(0)- f(1-t))\big) \pw(\dt)}{\int_{[0, 1]}u'\big(\ol{F_{X}^{-1}(\cdot)+\frac{1}{2\varc}( f(0)-f(\cdot))}-F_{X}^{-1}(1-t)-\tfrac{1}{2\varc}( f(0)- f(1-t))\big) \pw(\dt)}.
\end{align}
Note that $\oplc{f(\cdot)}$ can be regarded as a probability distribution function which may be discontinuous at the mass or singular points of $\pw$. If $\pw$ has a density, then so is $\oplc{f(\cdot)}$.

Now introduce the following OIDE of one unknown $\oplb$: 
\begin{equation} \label{vi002}
\begin{cases}
\min\Big\{\max\big\{-\oplb''(p), \; \oplb(p)-\oplc{\oplb'(\cdot)}(p)\big\}, \; 2\varc h(p)-\oplb''(p)
\Big\}=0, \hfill\quad\mbox{a.e.}\;p\in[0, 1], \bigskip\\
\oplb(0)=1-\vare, \quad \oplb(1)=1. 
\end{cases}
\end{equation} 
By virtue of \citelem{op2}, we can link it to the optimal solution to the problem \eqref{opi2}. The following simple technical result will be critical and used frequently in this process. 
\begin{lemma}[Lemma 4.4, \cite{X21}]\label{tech1}
Suppose $a, b, c$ are real numbers. Then $\min\{\max\{a, b\}, \;c\}=0$, if and only if $\min\{\max\{ak, b\ell\}, cm\}=0$ for any $k, \ell, m>0$. 
\end{lemma}

\begin{thm}[Optimal solution]\label{main1}
We have the following assertions. 
\begin{enumerate}[(1).]
\item 
If $\barg$ is the optimal solution to the problem \eqref{opi2}. Then
\begin{align} 
\oplb(p)&=\Big(2\varc \int_{0}^{1}\barg(t) \dt+\vare-2\varc\BE{X}\Big)(p-1)+2\varc\int_{p}^{1} \big(\barg(t)-F_{X}^{-1}(t)\big)\dt+1\label{oplbdef3}
\end{align} 
is a solution to \eqref{vi002} in $ C^{2-}([0, 1])$. 

\item 
If $\oplb$ is a solution to \eqref{vi002} in $ C^{2-}([0, 1])$. Then 
\begin{align}\label{bargdef1}
\barg(p)&=F_{X}^{-1}(p)+\frac{1}{2\varc}(\oplb'(0)-\oplb'(p)), 
\end{align}
and
\begin{align}
\overline{R}(x)&=\barg(F_X(x)) 
\end{align}
are optimal solutions to the problems \eqref{opi2} and \eqref{opi}, respectively. 
\end{enumerate}
As a consequence, \eqref{vi002} admits a unique solution in $ C^{2-}([0, 1])$. 
\end{thm}

\begin{proof}
\begin{enumerate}[(1).]
\item 
Since $\barg$ is absolutely continuous on $[0, 1]$, by the definition \eqref{oplbdef3} we have $\oplb\in C^{2-}([0, 1])$ and $\oplb(0)=1-\vare$ and $\oplb(1)=1$.
Differentiating \eqref{oplbdef3} and rearranging the terms, we get 
\[\barg(p)=F_{X}^{-1}(p)+\int_{0}^{1}\barg(t) \dt-\BE{X}+\frac{1}{2\varc}(\vare-\oplb'(p)).\]
Thanks to $\barg(0)=0$ and $F_{X}^{-1}(0)=0$, it follows
\[\barg(p)=F_{X}^{-1}(p)+\frac{1}{2\varc}(\oplb'(0)-\oplb'(p)).\]
Setting 
\begin{align}\label{opldef3}
\dsp\opl(p)&=\frac{\dsp\int_{(1-p, 1]}u'\big(\ol{\barg(\cdot)}-\barg(1-t)\big) \pw(\dt)}{\dsp\int_{[0, 1]}u'\big(\ol{\barg(\cdot)}-\barg(1-t)\big) \pw(\dt)}, 
\end{align} 
it yields 
\begin{align*} 
\dsp\opl(p) 
&=\frac{ \int_{(1-p, 1]} u'\big(\ol{F_{X}^{-1}(\cdot)+\frac{1}{2\varc}(\oplb'(0)-\oplb'(\cdot))}-F_{X}^{-1}(1-t)-\tfrac{1}{2\varc}(\oplb'(0)-\oplb'(1-t))\big) \pw(\dt)}{ \int_{[0, 1]}u'\big(\ol{F_{X}^{-1}(\cdot)+\frac{1}{2\varc}(\oplb'(0)-\oplb'(\cdot))}-F_{X}^{-1}(1-t)-\tfrac{1}{2\varc}(\oplb'(0)-\oplb'(1-t))\big) \pw(\dt)}\\
&=\oplc{\oplb'(\cdot)}(p).
\end{align*} 
As $\barg$ is the optimal solution to the problem \eqref{opi2}, thanks to \citelem{op2}, we have \eqref{vi001} which can be written as 
\begin{align*} 
\min\Big\{\max\big\{-\tfrac{1}{2\varc}\oplb''(p), \; \oplb(p)-\oplc{\oplb'(\cdot)}(p)\big\}, \; h(p)-\tfrac{1}{2\varc}\oplb''(p)
\Big\}=0.
\end{align*} 
By virtue of \citelem{tech1}, we deduce that $\oplb$ is a solution to \eqref{vi002} in $ C^{2-}([0, 1])$. 

\item 
Because $\oplb \in C^{2-}([0, 1])$, the definition \eqref{bargdef1} implies that $\barg$ is absolutely continuous, $\barg(0)=0$ and 
\begin{align*}
\barg'(p)=h(p)-\frac{1}{2\varc}\oplb''(p), \quad\mbox{a.e.}\;p\in[0, 1].
\end{align*} 
Setting $\opl(p)=\oplc{\oplb'(\cdot)}(p)$, it follows
\begin{align*} 
\dsp\opl(p)&=\frac{\int_{(1-p, 1]}u'\big(\ol{F_{X}^{-1}(\cdot)+\frac{1}{2\varc}(\oplb'(0)-\oplb'(\cdot))}-F_{X}^{-1}(1-t)-\tfrac{1}{2\varc}(\oplb'(0)-\oplb'(1-t))\big) \pw(\dt)}{\int_{[0, 1]}u'\big(\ol{F_{X}^{-1}(\cdot)+\frac{1}{2\varc}(\oplb'(0)-\oplb'(\cdot))}-F_{X}^{-1}(1-t)-\tfrac{1}{2\varc}(\oplb'(0)-\oplb'(1-t))\big) \pw(\dt)}\\
&=\frac{\dsp\int_{(1-p, 1]}u'\big(\ol{\barg(\cdot)}-\barg(1-t)\big) \pw(\dt)}{\dsp\int_{[0, 1]}u'\big(\ol{\barg(\cdot)}-\barg(1-t)\big) \pw(\dt)}, 
\end{align*}
which confirms \eqref{opldef}. We now show \eqref{vi001} and \eqref{oplbdef} are also satisfied. By virtue of \citelem{tech1}, the OIDE in \eqref{vi002} can be written as 
\begin{align*} 
\min\Big\{\max\big\{\barg'(p)-h(p), \; \oplb(p)-\opl(p)\big\}, \;\barg'(p)\Big\}=0, 
\end{align*} 
proving \eqref{vi001}. 
It follows from $\oplb(1)=1$ and \eqref{bargdef1} that 
\begin{align*} 
\oplb(p)&=c(p-1)+2\varc\int_{p}^{1} \big(\barg(t)-F_{X}^{-1}(t)\big)\dt+1
\end{align*} 
for some constant $c$. This together with the boundary condition $\oplb(0)=1-\vare$ confirms \eqref{oplbdef}. 
Now applying \citelem{op2}, we conclude that $\barg$ is an optimal solution to the problem \eqref{opi2}.
Consequently, $\overline{R}$ is the optimal solution to the problem \eqref{opi} by \eqref{RG}. 
\end{enumerate} 
Suppose $\oplb_1$ and $\oplb_2$ are two solutions in $ C^{2-}([0, 1])$ to the OIDE \eqref{vi002}. We then get two optimal solutions to the problem \eqref{opi2} from the relationship \eqref{bargdef1}. But by \citelem{lemexist1}, the optimal solution to \eqref{opi2} is unique, so we conclude $\oplb'_1(0)-\oplb'_1(p)=\oplb'_2(0)-\oplb'_2(p)$ for $p\in[0, 1]$, which implies that $\oplb_1-\oplb_2$ is a linear function. Because $\oplb_1(0)-\oplb_2(0)=0$ and $\oplb_1(1)-\oplb_2(1)=0$, we conclude that $\oplb_1-\oplb_2$ is identical to zero. Therefore, the OIDE \eqref{vi002} admits a unique solution in $ C^{2-}([0, 1])$. 
\end{proof} 

One may wonder if we always have a classical $C^2$ solution to the OIDE \eqref{vi002}. Generally speaking, this is not true. Because the optimal retention to the problem \eqref{opi}, $\overline{R}$ (such as deductible contracts), may not be a $C^1$ function, so is the optimal solution to the problem \eqref{opi2}. From \eqref{bargdef1}, we conclude that the OIDE \eqref{vi002} may not have a $C^2$ function. In the following section, we give an example to show this fact. 

\subsection{An example with explicit solution.}\label{example}
In this section, we provide an example whose solution will be derived from \citethm{main1} explicitly. In this example, the probability measure $\pw$ is highly nonlinear, so the problem \eqref{opi} is indeed a behavioral model (see \citeremark{remark:bhm}), resulting in a non-classical PO insurance contract. 

Let 
\begin{align*}
F_X(x)=\frac{x+1}{2},\;x\in[0,1]; \quad \theta=\frac{13}{12}, \quad \sigma=\frac{1}{2}, \quad u(x)=-e^{-x}.
\end{align*}
Let the probability measure $\pw:[0,1]\to[0,1]$ to be an absolutely continuous function that satisfies the following condition: 
\begin{align*} 
\pw'(1-p)
&=\begin{cases}
ce^{\frac{8}{9}} \big(2p-\frac{1}{9}\big),&\quad p\in\big(\frac{1}{2},\frac{2}{3}\big)\bigskip\\
c\big(p+\frac{5}{9}\big)\exp\big(-p+\frac{14}{9}\big), &\quad p\in\big(\frac{2}{3},1\big),
\end{cases}
\end{align*}
where
\begin{align*} 
c^{-1}&=\frac{13}{36}e^{-\frac{8}{9}}+\int_{[\frac{1}{2},\frac{2}{3}]}e^{\frac{8}{9}}\big(2t-\frac{1}{9}\big)\dt+\int_{[\frac{2}{3},1]}\big(t+\frac{5}{9}\big)\exp\big(-t+\frac{14}{9}\big)\dt.
\end{align*} 
We do not put any constraint on $\pw'(1-p)$, $p\in (0,\frac{1}{2})$. Also, because $u$ is an exponential utility, the wealth position $\wa-\gamma$ of the insured does not affect the optimal contract, so we do not give them explicitly.

It is easy to verify that 
\begin{align*}
\BE{X}=\frac{1}{4}, \quad F_X^{-1}(p)=(2p-1)^+,\quad h(p)=2\idd{p\in (\frac{1}{2},1]}, 
\end{align*}
where $\idd{S}$ is the indicator function for a statement $S$, so $\idd{S}=1$ if $S$ is true and $\idd{S}=0$ otherwise. 

If we can show that 
\begin{align}\label{ep1} 
\oplb(p)&=\frac{8}{9}p-\frac{1}{12}-\frac{1}{2}\big(p-\frac{2}{3}\big)^2\idd{p\in [\frac{2}{3},1]}
+\big(p-\frac{1}{2}\big)^2\idd{p\in [\frac{1}{2},1]}, \quad p\in[0,1],
\end{align}
is a solution to the OIDE \eqref{vi002} in $ C^{2-}([0, 1])$ (but it clearly does not belong to $C^{2}([0, 1])$). 
It then follows from \citethm{main1} that 
\begin{align*} 
\barg(p)&=F_{X}^{-1}(p)+\frac{1}{2\varc}(\oplb'(0)-\oplb'(p))=\big(p-\tfrac{2}{3}\big)^+, 
\quad p\in[0,1],
\end{align*}
is the optimal solution to the problem \eqref{opi2} and 
\begin{align*}
\overline{R}(x)&=\barg(F_X(x))=\frac{1}{6}(3x-1)^+,\quad x\in[0,1],
\end{align*}
is an optimal retention to the problem \eqref{opi}. This contract is neither a classical PO deductible, nor a proportional insurance contract. It is indeed because the nonlinear probability measure $\pw$ makes the problem \eqref{opi} a behavioral model under probability distortion (see \citeremark{remark:bhm}). 

We now show that $\oplb$ defined by \eqref{ep1} is a solution to \eqref{vi002}.
It is evident that $\oplb(0)=-\frac{1}{12}=1-\vare$ and $\oplb(1)=1$, so the two boundary conditions in \eqref{vi002} are satisfied. 
Also, trivially, 
\begin{multline*} 
\min\Big\{\max\big\{-\oplb''(p), \; \oplb(p)-\oplc{\oplb'(\cdot)}(p)\big\}, \; 2\varc h(p)-\oplb''(p)\Big\}\\
=\min\Big\{\max\big\{0, \; \oplb(p)-\oplc{\oplb'(\cdot)}(p)\big\}, \; 0\Big\}=0, \quad p\in\big(0,\tfrac{1}{2}\big), 
\end{multline*} 
so $\oplb$ satisfies the OIDE in \eqref{vi002} on $(0,\frac{1}{2})$. 
If we can now show 
\begin{align}\label{example1}
\oplb(p)-\oplc{\oplb'(\cdot)}(p)=0,\quad p\in\big(\tfrac{1}{2},1\big),
\end{align} 
then it follows that 
\begin{multline*} 
\min\Big\{\max\big\{-\oplb''(p), \; \oplb(p)-\oplc{\oplb'(\cdot)}(p)\big\}, \; 2\varc h(p)-\oplb''(p)\Big\}\\
=\min\Big\{\max\big\{-2, \; 0\big\}, \; 2-2\Big\}=\min\big\{ 0, \; 0\big\}=0, 
\quad p\in\big(\tfrac{1}{2},\tfrac{2}{3}\big),
\end{multline*}
and 
\begin{multline*} 
\min\Big\{\max\big\{-\oplb''(p), \; \oplb(p)-\oplc{\oplb'(\cdot)}(p)\big\}, \; 2\varc h(p)-\oplb''(p)\Big\}\\
=\min\Big\{\max\big\{-1, \; 0\big\}, \; 2-1\Big\}=\min\big\{0, \; 1\big\}=0,\quad p\in\big(\tfrac{2}{3},1\big),
\end{multline*}
so we can conclude that $\oplb$ is a solution to \eqref{vi002}. 

It suffices to prove \eqref{example1}.
Because $\oplb(1)=1=\oplc{\oplb'(\cdot)}(1)$, we only need to prove 
\[\frac{\dd}{\ddp}\oplc{\oplb'(\cdot)}(p)=\oplb'(p),\quad p\in\big(\tfrac{1}{2},1\big).\]
Because $u(x)=-e^{-x}$, we have 
\begin{align*} 
\oplc{\oplb'(\cdot)}(p)&=\frac{ \int_{(1-p, 1]} u'\big(\ol{F_{X}^{-1}(\cdot)+\frac{1}{2\varc}(\oplb'(0)-\oplb'(\cdot))}-F_{X}^{-1}(1-t)-\tfrac{1}{2\varc}(\oplb'(0)-\oplb'(1-t))\big) \pw(\dt)}{ \int_{[0, 1]}u'\big(\ol{F_{X}^{-1}(\cdot)+\frac{1}{2\varc}(\oplb'(0)-\oplb'(\cdot))}-F_{X}^{-1}(1-t)-\tfrac{1}{2\varc}(\oplb'(0)-\oplb'(1-t))\big) \pw(\dt)}\\
&=\frac{ \int_{(1-p, 1]} \exp\big(F_{X}^{-1}(1-t)-\oplb'(1-t)\big) \pw'(t)\dt}
{\int_{[0, 1]}\exp\big(F_{X}^{-1}(1-t)-\oplb'(1-t)\big) \pw'(t)\dt}\\
&=d^{-1}\int_{[0,p)} \exp\big(F_{X}^{-1}(t)-\oplb'(t)\big) \pw'(1-t)\dt,
\end{align*} 
where 
\[d=\int_{[0, 1]}\exp\big(F_{X}^{-1}(1-t)-\oplb'(1-t)\big) \pw'(t)\dt
=\int_{[0, 1]}\exp\big(F_{X}^{-1}(t)-\oplb'(t)\big) \pw'(1-t)\dt.\]
It follows that 
\begin{align*}
\frac{\dd}{\ddp}\oplc{\oplb'(\cdot)}(p)&=d^{-1}\exp\big(F_{X}^{-1}(p)-\oplb'(p)\big)\pw'(1-p)\\
&=d^{-1}\exp\big(2p-1-\big(2p-\frac{1}{9}\big)\big)ce^{\frac{8}{9}} \big(2p-\frac{1}{9}\big)\\
&=cd^{-1}\big(2p-\frac{1}{9}\big)\\
&=cd^{-1}\oplb'(p), \qquad\qquad\qquad p\in\big(\tfrac{1}{2},\tfrac{2}{3}\big),
\end{align*}
and 
\begin{align*}
\frac{\dd}{\ddp}\oplc{\oplb'(\cdot)}(p)&=d^{-1}\exp\big(F_{X}^{-1}(p)-\oplb'(p)\big)\pw'(1-p)\\
&=d^{-1}\exp\big(2p-1-\big(p+\frac{5}{9}\big)\big) c\big(p+\frac{5}{9}\big)\exp\big(-p+\frac{14}{9}\big)\\
&=cd^{-1}\big(p+\frac{5}{9}\big)\\
&=cd^{-1}\oplb'(p), \qquad\qquad\qquad p\in\big(\tfrac{2}{3},1\big).
\end{align*}
Hence, it suffices to show $c=d$. Using the above equations, we see
\begin{align*} 
d&=\int_{[0, 1]}\exp\big(F_{X}^{-1}(t)+\oplb'(t)\big) \pw'(1-t)\dt\\ 
&=\int_{[0, \frac{1}{2}]}\exp\big(F_{X}^{-1}(t)+\oplb'(t)\big) \pw'(1-t)\dt\\
&\qquad\;+\int_{(\frac{1}{2}, 1]}\exp\big(F_{X}^{-1}(t)+\oplb'(t)\big) \pw'(1-t)\dt\\ 
&=\int_{[0, \frac{1}{2}]}e^{\frac{8}{9}} \pw'(1-t)\dt+\int_{(\frac{1}{2}, 1]}c\oplb'(t)\dt\\ 
&=e^{\frac{8}{9}} (1-\pw([0,1/2]))+c(1-\oplb(1/2))\\
&=e^{\frac{8}{9}} (1-\pw([0,1/2]))+\frac{23}{36}c.
\end{align*} 
Also, it follows from definitions of $c$ and $\mu'$ that 
\begin{align*} 
1-\frac{13}{36}e^{-\frac{8}{9}}c 
&=c\int_{[\frac{1}{2},\frac{2}{3}]}e^{\frac{8}{9}}\big(2t-\frac{1}{9}\big)\dt
+c\int_{[\frac{2}{3},1]}\big(t+\frac{5}{9}\big)\exp\big(-t+\frac{14}{9}\big)\dt\\
&=\int_{[\frac{1}{2},1]}\pw'(1-t)\dt=\pw([0,1/2]).
\end{align*} 
As a by-product, it implies $0<\pw([0,1/2])<1$, so $\pw$ is a well-defined probability measure. 
Moreover, the above two equations lead to 
\begin{align*} 
d &=e^{\frac{8}{9}} (1-\pw([0,1/2]))+\frac{23}{36}c=\frac{13}{36}c+\frac{23}{36}c=c.
\end{align*} 
This completes the proof.

\subsection{Special case: when $\pw$ has a density.}

Generally speaking, the OIDE problem \eqref{vi002} is not easy to solve, even numerically, because the operator $\oplc{\oplb'(\cdot)}$ is non-local. But when $\pw$ has a density, we can further reduce the OIDE problem to an initial value ODE problem of two unknown functions that can be effectively numerically solved. 

\begin{thm}[Optimal solution when $\pw$ has a density]
Suppose the probability measure $\pw$ has a density $\pw'$. Then $\oplb$ is the solution to \eqref{vi002} in $ C^{2-}([0, 1])$ if and only if it can be expressed as 
\begin{align} \label{oplbdef2}
\oplb(p)=2\varc\int_0^p (F_{X}^{-1}(t)-\opld(t))\dt+\varf p+1-\vare, 
\end{align}
where $(\ople, \opld, c, d, \varf)$ is the unique solution to the following ODE system 
\begin{equation} \label{vi002a}
\begin{cases}
\min\Big\{\max\big\{\opld'(p)-h(p), \; \ople(p)\big\}, \;\opld'(p)\Big\}=0, \qquad\bigskip\\
\ople'(p)=2\varc(F_{X}^{-1}(p)-\opld(p))+\varf-c u'(d-\opld(p)) \pw'(1-p), 
\hfill\quad\mbox{a.e.}\;p\in[0, 1], \bigskip\\
\ople(0)=1-\vare, \; \ople(1)=0, \; \opld(0)=0, \;\ol{\opld(\cdot)}=d, \; \varf=\vare+2\varc\int_0^1 \opld(t)\dt-2\varc \BE{X}, 
\end{cases}
\end{equation} 
in the sense that $\ople$, $\opld\in \setac([0, 1])$, and $(c, d, \varf)\in(0, \infty)\times\R\times\R$. 
Moreover, $\opld$ is the optimal solution to the problem \eqref{opi2}. 
\end{thm}
\begin{proof}
Suppose $\oplb$ is the solution to \eqref{vi002} in $ C^{2-}([0, 1])$. Set 
\begin{align*} 
\dsp\ople(p)&=\oplb(p)-\oplc{\oplb'(\cdot)}(p), \\
\opld(p)&=F_{X}^{-1}(p)+\frac{1}{2\varc}(\oplb'(0)-\oplb'(p)), \\
c^{-1}&=\int_{[0, 1]}u'\big(\ol{F_{X}^{-1}(\cdot)+\frac{1}{2\varc}(\oplb'(0)-\oplb'(\cdot))}-\opld(1-p) \big) \pw(\dt), \\
d&=\ol{F_{X}^{-1}(\cdot)+\frac{1}{2\varc}(\oplb'(0)-\oplb'(\cdot))}, \\
\varf &=\oplb'(0).
\end{align*}
By virtue of \citelem{tech1}, it is easy to check that $(\ople, \opld, c, d, \varf)$ is a solution to \eqref{vi002a} except for the last boundary condition
\[\varf=\vare+2\varc\int_0^1 \opld(t)\dt-2\varc \BE{X}.\]
Integrating both sides of 
\[\opld(p)=F_{X}^{-1}(p)+\frac{1}{2\varc}(\varf-\oplb'(p))\]
over $[0,1]$ and using $\oplb(1)-\oplb(0)=\vare$, we get the last boundary condition in \eqref{vi002a}.
By \eqref{bargdef1}, we see that $\opld$ is the optimal solution to the problem \eqref{opi2}. 

To show the reverse implication, 
we suppose $(\ople, \opld, c, d, \varf)$ is a solution to \eqref{vi002a}. Set $\oplb$ by \eqref{oplbdef2}, then $ \oplb\in C^{2-}([0, 1])$ and $\oplb(1)=1$ by virtue of the last boundary condition of \eqref{vi002a}. 
Thanks to \eqref{oplbdef2} and the boundary condition $\opld(0)=0$, we get
\begin{align*} 
\opld(p)&=F_{X}^{-1}(p)+\frac{1}{2\varc}(\oplb'(0)-\oplb'(p)).
\end{align*} 
So
\begin{align*} 
\opld'(p)&=h(p)-\frac{1}{2\varc}\oplb''(p), 
\end{align*} 
and
\[d=\ol{\opld(\cdot)}=\ol{F_{X}^{-1}(\cdot)+\frac{1}{2\varc}(\oplb'(0)-\oplb'(\cdot))}.\]
Writing the second ODE in \eqref{vi002a} as 
\[\ople'(p)=\oplb'(p)-c u'(d-\opld(p)) \pw'(1-p),\]
and using $\oplb(0)=\ople(0)$, we obtain 
\begin{multline*} 
\oplb(p)-\ople(p)=c \int_{(1-p, 1]}u'(d-\opld(1-t)) \pw'(t)\dt\\
=c\int_{(1-p, 1]}u'\big(\ol{F_{X}^{-1}(\cdot)+\frac{1}{2\varc}(\oplb'(0)-\oplb'(\cdot))}-F_{X}^{-1}(1-t)-\tfrac{1}{2\varc}(\oplb'(0)-\oplb'(1-t))\big) \pw(\dt). 
\end{multline*}
By virtue of $\oplb(1)-\ople(1)=1$, it yields 
\begin{align*} 
c^{-1}&=\int_{[0, 1]}u'\big(\ol{F_{X}^{-1}(\cdot)+\frac{1}{2\varc}(\oplb'(0)-\oplb'(\cdot))}-F_{X}^{-1}(1-t)-\tfrac{1}{2\varc}(\oplb'(0)-\oplb'(1-t))\big) \pw(\dt)>0, 
\end{align*}
and consequently, \[\oplb(p)-\ople(p)=\oplc{\oplb'(\cdot)}(p).\] Replacing $\ople$ and $\opld$ in the first ODE in \eqref{vi002a}, we conclude that $\oplb$ is a solution to \eqref{vi002} in $ C^{2-}([0, 1])$. By \eqref{bargdef1}, $\opld$ is the optimal solution to the problem \eqref{opi2}. 
Because \eqref{vi002} admits a unique solution in $ C^{2-}([0, 1])$, the above equivalency shows that \eqref{vi002a} admits a unique solution. 
\end{proof}

Finally, we propose a numerical scheme to solve \eqref{vi002a}. 
To solve \eqref{vi002a}, one can try, for each triple $(c, d, \varf)\in(0, \infty)\times\R\times\R$, to get the numerical solution to the following initial value problem
\begin{equation*} 
\begin{cases}
\min\Big\{\max\big\{\opld'(p)-h(p), \; \ople(p)\big\}, \;\opld'(p)\Big\}=0, \qquad\bigskip\\
\ople'(p)=2\varc(F_{X}^{-1}(p)-\opld(p))+\varf-c u'(d-\opld(p)) \pw'(1-p), 
\hfill\quad\mbox{a.e.}\;p\in[0, 1], \bigskip\\
\ople(0)=1-\vare, \quad \opld(0)=0, 
\end{cases}
\end{equation*} 
until its solution $(\ople,\opld)$ satisfies $\ople(1)=0$, $\ol{\opld(\cdot)}=d$ and $\varf=\vare+2\varc\int_0^1 \opld(t)\dt-2\varc \BE{X}$. Then $(\ople, \opld, c, d, \varf)$ is the solution to \eqref{vi002a}.

The above scheme can be further optimized. For instance, because $\opld$ is optimal to the problem \eqref{opi2}, $\opld$ is bounded in $[0, \esup X]$ by virtue of \eqref{boundedsetg}. Using this fact, one can establish explicit bounds for $c$, $d$, and $\varf$. 
Therefore, one just needs to look for the desired tuple $(c, d, \varf)$ in a known bounded region rather than in the whole space $(0, \infty)\times\R\times\R$. 

\par 

\section{Concluding remarks.}\label{conclusion}
In this paper, we have considered the PO moral-hazard-free insurance contract problem for an insurer using the mean-variance premium principle and an insured using rank-dependent utility preference. We have mainly focused on the theoretical study of the problem and proposed a numerical scheme to solve the associated ODE problem. We believe the most up-to-date numerical methods to solve differential equations such as neural networks and deep learning might be applied to our OIDE and ODE problems. We encourage experts in the relevant fields to study them. We also encourage the interested readers to extend our model to incorporate other premium principles such as the standard deviation premium principle. We believe our method should be effective whenever the premium principle has certain concavity. 




\begin{thebibliography}{99}
\bibitem{A53}\textsc{Allais, M. (1953):}
Le comportement de l'homme rationnel devant le risque: critique des postulats et axiomes de l'ecole americaine, 
\textit{Econometrica}, Vol. 21(4), pp. 503-546.

\bibitem{BHYZ15}\textsc{Bernard, C., He, X.D., Yan, J.-A., and Zhou, X.Y. (2015):}
Optimal insurance design under rank-dependent expected utility, 
\textit{Mathematical Finance}, Vol. 25, pp. 154-186.

\bibitem{CD08}\textsc{Carlier, G., and Dana, R.-A. (2008):}
Two-persons efficient risk-sharing and equilibria for concave law-invariant utilities, 
\textit{Economic Theory}, Vol. 36(2), pp. 189-223.

\bibitem{CL11}\textsc{Carlier, G., and Lachapelle, A. (2011):}
A numerical approach for a class of risk-sharing problems,
\textit{Journal of Mathematical Economics}, Vol. 47, pp. 1-13.

\bibitem{CDT00}\textsc{Chateauneuf, A., Dana, R.-A., and Tallon. J.-M. (2000):}
Optimal risk-sharing rules and equilibria with Choquet-expected-utility, 
\textit{Journal of Mathematical Economics }, Vol. 34(2), pp. 191-214.



\bibitem{DS07}\textsc{Dana, R.-A., and Scarsini, M. (2007):}
Optimal risk sharing with background risk, 
\textit{Journal of Economic Theory}, Vol. 133(1), pp. 152-176.


\bibitem{DPP94}\textsc{Daykin, C.D., Pentikainen, T., and Pesonen, M. (1994).} 
\textit{Practical Risk Theory for Actuaries}, Chapman \& Hall, London.

\bibitem{DG85}\textsc{Deprez, O., and Gerber, U. (1985):}
On convex principles of premium calculation,
\textit{Insurance: Mathematics and Economics}, Vol. 4, pp. 179-189. 


\bibitem{E61}\textsc{Ellsberg, D. (1961):}
Risk, ambiguity and the Savage axioms, 
\textit{Quarterly Journal of Economics}, Vol. 75(4), pp. 643-669.

\bibitem{FS48}\textsc{Friedman, M., and Savage, L.J. (1948):}
The utility analysis of choices involving risk, 
\textit{Journal of Political Economy}, Vol. 56(4), pp. 279-304.



\bibitem{GZ00}\textsc{Gajek, L., Zagrodny, D. (2000):} 
Insurer’s optimal reinsurance strategies. 
\textit{Insurance: Mathematics and Economics}, Vol. 27, pp. 105-112.


\bibitem{GXZ22}\textsc{Guan, C., Xu, Z. Q., and Zhou, R. (2022):}
Dynamic optimal reinsurance and dividend-payout in finite time horizon,
to appear in \textit{Mathematics of Operations Research}, \url{https://doi.org/10.1287/moor.2022.1276} and \url{https://arxiv.org/abs/2008.00391}

\bibitem{HJZ15} \textsc{He, X. D., Jin, H., and Zhou, X. Y. (2015):}
Dynamic portfolio choice when risk is measured by Weighted VaR, 
\textit{Mathematics of Operations Research}, Vol. 40, pp. 773-796.

\bibitem{HZ11} \textsc{He, X. D., and X. Y. Zhou (2011):}
Portfolio choice via quantiles, 
\textit{Mathematical Finance}, Vol. 21, pp. 203-231.

\bibitem{HT10} \textsc{Hipp, C., Taksar, M, (2010):}
Optimal non-proportional reinsurance, 
\textit{Insurance Math. Econom.}, Vol. 47 (2), pp. 246-254.

\bibitem{HX16} \textsc{Hou, D., and Z. Q. Xu (2016):} 
A robust Markowitz mean--variance portfolio selection model with an intractable claim, 
\textit{SIAM Journal on Financial Mathematics}, Vol.7, 124-151.


\bibitem{HMS83} \textsc{Huberman, G., Mayers, D., and Smith Jr, C.W. (1983):}
Optimal insurance policy indemnity schedules, 
\textit{The Bell Journal of Economics}, Vol. 14(2), pp. 415-426.

\bibitem{JZ08} \textsc{Jin, H., and X. Y. Zhou (2008):}
Behavioral portfolio selection in continuous time, 
\textit{Mathematical Finance}, Vol. 18, pp. 385-426

\bibitem{KT79}\textsc{Kahneman, D., and A. Tversky (1979):}
Prospect theory: An analysis of decision under risk, 
\textit{Econometrica}, Vol. 46, pp. 171-185

\bibitem{K01}\textsc{Kaluszka, M. (2001):}
Optimal reinsurance under mean-variance premium principles,
\textit{Insurance: Mathematics and Economics}, Vol. 28, pp. 61-67. 


\bibitem{LLY20}\textsc{Liang, X., Liang, Z, and Young, V. (2020):}
Optimal reinsurance under the mean-variance premium principle to minimize the probability of ruin, 
\textit{Insurance: Mathematics and Economics}, Vol. 92, pp. 128-146. 


\bibitem{L87}\textsc{Lopes, L. L. (1987):} 
Between hope and fear: The psychology of risk, 
\textit{Advances in experimental social psychology}, Vol. 20, pp.255-295 



\bibitem{MP85}\textsc{Mehra, R., and Prescott, E.C. (1985):}
The equity premium: A puzzle, 
\textit{Journal of Monetary Economics}, Vol. 15(2), pp. 145-161

\bibitem{MX21}\textsc{Mi, H., and X, Z.Q. (2021):}
Optimal portfolio selection with VaR and portfolio insurance constraints
under rank-dependent expected utility theory, working paper, \url{http://ssrn.com/abstract=3880289}.


\bibitem{P00}\textsc{Picard, P. (2000):}
On the design of optimal insurance policies under manipulation of audit cost, 
\textit{International Economic Review}, Vol. 41(4), pp. 1049-1071. 

\bibitem{Q82}\textsc{Quiggin (1982):}
A theory of anticipated utility, 
\textit{Journal of Economic and Behavioral Organization}, Vol. 3(4), pp. 323-343


\bibitem{TK92} \textsc{Tversky, A., and D. Kahneman (1992):}
Advances in prospect theory: Cumulative representation of uncertainty, 
\textit{J. Risk Uncertainty}, Vol. 5, pp. 297-323


\bibitem{W18} \textsc{Wei, P. (2018):} 
Risk management with weighted VaR, 
\textit{Mathematical Finance}, Vol. 28(4), 1020-1060



\bibitem{XZ16}	\textsc{Xia, J. M., and X. Y. Zhou (2016):} 
Arrow-Debreu equilibria for rank-dependent utilities, 
\textit{Mathematical Finance}, Vol. 26, pp. 558-588

\bibitem{X14}	\textsc{Xu, Z. Q. (2014):}
A new characterization of comonotonicity and its application in behavioral finance, 
\textit{J. Math. Anal. Appl.}, Vol. 418, pp. 612-625


\bibitem{X16}	\textsc{Xu, Z. Q. (2016):}
A note on the quantile formulation, 
\textit{Mathematical Finance}, Vol.26, No. 3, 589-601

\bibitem{X21}	\textsc{Xu, Z. Q. (2021):}
Pareto optimal moral-hazard-free insurance contracts in behavioral finance framework, 
working paper, \url{https://arxiv.org/abs/1803.02546}



\bibitem{XZ13}	\textsc{Xu, Z. Q., and X. Y. Zhou (2013):}
Optimal stopping under probability distortion, 
\textit{Annals of Applied Probability}, Vol. 23, pp. 251-282


\bibitem{XZZ19}	\textsc{Xu, Z. Q., X. Y. Zhou, and S. Zhuang (2019):}
Optimal insurance under rank-dependent utility and incentive compatibility, 
\textit{Mathematical Finance}, Vol. 29(2), pp. 659-692


\bibitem{Y87}\textsc{Yaari, M.E. (1987):}
The dual theory of choice under risk, 
\textit{Econometrica}, Vol. 55(1), pp. 95-115


\bibitem{YYW14}\textsc{Yao, D., Yang, H., and Wang, R. (2014):}
Optimal risk and dividend control problem with fixed costs and salvage value: Variance premium principle
\textit{Economic Modelling}, Vol. 37, pp. 53-64. 


\end{thebibliography}
\end{document}